\documentclass[a4paper,english,nosubfigcap]{lipics-v2019}

\hideLIPIcs
\nolinenumbers

\usepackage[utf8]{inputenc}

\usepackage{xspace}
\usepackage{amsmath}
\usepackage{amssymb}
\usepackage{mathtools}
\usepackage{csquotes}

\newcommand{\abs}[1]{\ensuremath{\lvert #1 \rvert}}
\newcommand{\norm}[1]{\ensuremath{\lVert #1 \rVert}}
\newcommand{\dhausdorff}{\ensuremath{\delta_{\vec{H}}}}
\newcommand{\dhausdorfft}{\ensuremath{\delta_{\vec{H}}^T}}
\newcommand{\hausdorff}{\ensuremath{\delta_H}}
\newcommand{\hausdorfft}{\ensuremath{\delta_H^T}}

\newcommand{\convthreesum}{\textsc{Conv3Sum}\xspace}
\newcommand{\threesum}{\textsc{3Sum}\xspace}

\newcommand{\Oh}{\ensuremath{\mathcal{O}}\xspace}
\newcommand{\Ohtilda}{\ensuremath{\tilde{\mathcal{O}}}\xspace}
\newcommand{\lone}{\ensuremath{L_1}\xspace}
\newcommand{\ltwo}{\ensuremath{L_2}\xspace}
\newcommand{\linf}{\ensuremath{L_\infty}\xspace}
\newcommand{\lp}{\ensuremath{L_p}\xspace}
\newcommand{\RR}{\mathbb{R}\xspace}

\renewcommand{\epsilon}{\varepsilon}
\newcommand{\eps}{\varepsilon}

\title{Translating Hausdorff is Hard: Fine-Grained Lower Bounds for Hausdorff Distance Under Translation}
\titlerunning{Fine-Grained Lower Bounds for Hausdorff Distance Under Translation}

\author{Karl Bringmann}{Saarland University and Max Planck Institute for Informatics,  Saarland Informatics Campus, Saarbrücken, Germany}{bringmann@cs.uni-saarland.de}{}{This work is part of the project TIPEA that has received funding from the European Research Council (ERC) under the European Unions Horizon 2020 research and innovation programme (grant agreement No. 850979).}
\author{André Nusser}{Saarbrücken Graduate School of Computer Science and Max Planck Institute for Informatics,  Saarland Informatics Campus, Saarbrücken, Germany}{anusser@mpi-inf.mpg.de}{}{Part of this work was done at BARC, Copenhagen University, supported by the VILLUM Foundation grant 16582.}

\authorrunning{K. Bringmann, A. Nusser}

\Copyright{Karl Bringmann and André Nusser}

\ccsdesc[500]{Theory of computation~Problems, reductions and completeness} 

\keywords{Hausdorff Distance Under Translation, Fine-Grained Complexity Theory, Lower Bounds}

\begin{document}
\maketitle

\begin{abstract}
Computing the similarity of two point sets is a ubiquitous task in medical imaging, geometric shape comparison, trajectory analysis, and many more settings.
Arguably the most basic distance measure for this task is the Hausdorff distance, which
assigns to each point from one set the closest point in the other set and then evaluates the maximum distance of any assigned pair.
A drawback is that this distance measure is not translational invariant, that is, comparing two objects just according to their shape while disregarding their position in space is impossible. 

Fortunately, there is a canonical translational invariant version, the Hausdorff distance under translation, which minimizes the Hausdorff distance over all translations of one of the point sets.
For point sets of size $n$ and $m$, the Hausdorff distance under translation can be computed in time $\Ohtilda(nm)$ for the $\lone$ and $\linf$ norm [Chew, Kedem SWAT'92] and $\Ohtilda(nm (n+m))$ for the $\ltwo$ norm [Huttenlocher, Kedem, Sharir DCG'93].

As these bounds have not been improved for over 25 years, in this paper we approach the Hausdorff distance under translation from the perspective of fine-grained complexity theory.
We show (i) a matching lower bound of $(nm)^{1-o(1)}$ for \lone and \linf (and all other \lp norms) assuming the Orthogonal Vectors Hypothesis and (ii) a matching lower bound of $n^{2-o(1)}$ for \ltwo in the imbalanced case of $m = \Oh(1)$ assuming the 3SUM Hypothesis.  \end{abstract}

\section{Introduction} \label{sec:introduction}
As data sets become larger and larger, the requirement for faster algorithms to handle such amounts of data becomes increasingly necessary. One very common type of data that is created during measurements is point sets in the plane, for example when recording GPS trajectories or describing shapes of objects, in medical image analysis, and in various data science applications.

A fundamental algorithmic tool for analyzing point sets is to compute the similarity of two given sets of points. There are several different measures of similarity in this setting, for example Hausdorff distance \cite{hausdorff1914grundzuge}, geometric bottleneck matching \cite{efrat_geometry_2001}, Fréchet distance~\cite{alt_computing_1995}, and Dynamic Time Warping \cite{DBLP:books/daglib/0019158}. 
Among these measures, the Hausdorff distance is arguably the most basic and intuitive: It assigns to each point from one set the closest point in the other set and then evaluates the maximum distance of all assigned pairs of points.\footnote{There is a directed and an undirected variant of the Hausdorff distance, see Section~\ref{sec:preliminaries}. In this introduction, we do not differentiate between these two, since all our statements hold for both variants.} 
For a discussion of the other previously mentioned distance measures, see Section~\ref{sec:related_work}.

While these similarity measures are of great practical relevance, for some applications it is a drawback that they are not 
translational invariant, i.e., when translating one of the point sets, the distance can -- and in most cases will -- change. 
This is unfavorable in applications that ask for comparing the shape of two objects, meaning that the absolute position of an object is irrelevant. Examples of this task arise for example in 2D object shape similarity, medical image analysis \cite{fedorov2008evaluation}, classification of handwritten characters~\cite{DBLP:conf/esa/BringmannKN20}, movement patterns of animals \cite{DBLP:conf/gis/BuchinDLN19}, and sports analysis \cite{DEBERG2013747}.

Fortunately, any point set similarity measure has a canonical translational invariant version, by minimizing the similarity measure over all translations of the two given point sets.
For the Hausdorff distance this variant is known as the \emph{Hausdorff distance under translation}, see Section~\ref{sec:preliminaries} for a formal definition. Given two point sets in the plane of size $n$ and $m$, the Hausdorff distance under translation can be computed in time $\Oh(nm \log^2 nm)$ for the $\lone$ and $\linf$ norm~\cite{n2alg}, and in time $\Oh(nm (n+m) \log nm)$ for the $\ltwo$ norm~\cite{n3alg}. We are not aware of any lower bounds for this problem, not even conditional on a plausible hypothesis.
The only results in this direction are $\Omega(n^3)$ lower bounds on the arrangement size~\cite{n2alg} and on the number of connected components of the feasible translations~\cite{components_lb} (for the decision problem on points in the plane with $n=m$). However, these bounds also hold for \lone and \linf, where they are ``broken'' by the $\Oh(nm \log^2 nm)$-time algorithm~\cite{n2alg}, so apparently these bounds are irrelevant for the running time complexity.

\medskip
In this paper, we approach the Hausdorff distance under translation from the viewpoint of fine-grained complexity theory~\cite{williams2018some}.
For two problem settings, we show that the known algorithms are optimal up to lower order factors assuming standard hypotheses:

\begin{enumerate}
	\item We show an $(nm)^{1-o(1)}$ lower bound for all \lp norms --- and in particular \lone and \linf, matching the $\Oh(nm \log^2 nm)$-time algorithm from~\cite{n2alg} up to lower order factors, see Section~\ref{sec:ov_lb}.
	
	This result holds conditional on the Orthogonal Vectors Hypothesis, which states that finding two orthogonal vectors among two given sets of $n$ binary vectors in $d$ dimensions cannot be done in time $\Oh(n^{2-\varepsilon} \textrm{poly}(d))$ for any $\varepsilon > 0$. It is well-known that the Orthogonal Vectors Hypothesis is implied by the Strong Exponential Time Hypothesis~\cite{DBLP:journals/tcs/Williams05}, and thus our lower bound also holds assuming the latter~\cite{DBLP:journals/jcss/ImpagliazzoPZ01}.
	These two hypotheses are the most standard assumptions used in fine-grained complexity theory in the last decade~\cite{williams2018some}.
	\item We show an $n^{2-o(1)}$ lower bound for $\ltwo$ in the imbalanced case $m = \Oh(1)$, matching the $\Oh(nm (n+m) \log nm)$-time algorithm from~\cite{n2alg} up to lower order factors, see Section~\ref{sec:3sum_lb}.
	Previously, an $n^{2-o(1)}$ lower bound was only known for the more general problem of computing the Hausdorff distance under translation of sets of \emph{segments} in the case that both sets have size $n$ (a problem for which the best known algorithm runs in time\footnote{By $\Ohtilda$-notation we ignore logarithmic factors in $n$ and $m$.} $\Ohtilda(n^4)$)~\cite{3sum}.
	
	Our result holds conditional on the 3SUM Hypothesis, which states that deciding whether, among $n$ given integers, there are three that sum up to 0 requires time $n^{2-o(1)}$. This hypothesis was introduced by Gajentaan and Overmars~\cite{DBLP:journals/comgeo/GajentaanO95}, is a standard assumption in computational geometry~\cite{king2004survey}, and has also found a wealth of applications beyond geometry (see, e.g.,~\cite{DBLP:conf/focs/AbboudBBK17,DBLP:conf/icalp/AbboudWW14,DBLP:conf/icalp/AmirCLL14,3sum_to_conv3sum}).
\end{enumerate}

Our lower bounds close gaps that have not seen any progress over 25 years. Furthermore, note that our second lower bound shows a separation between the $\ltwo$ norm and the $\lone$ and $\linf$ norms, as in the imbalanced case $m = \Oh(1)$ the latter admits a $\Ohtilda(n)$-time algorithm~\cite{n2alg} while the former requires time $n^{2-o(1)}$ assuming the 3SUM Hypothesis.
We leave it as an open problem whether for \ltwo the balanced case $n = m$ requires time $n^{3-o(1)}$.

\subsection{Related work} \label{sec:related_work}

Our work continues a line of research on fine-grained lower bounds in computational geometry, which had early success with the 3SUM Hypothesis~\cite{DBLP:journals/comgeo/GajentaanO95} and recently got a new impulse with the Orthogonal Vectors Hypothesis (or Strong Exponential Time Hypothesis) and resulting lower bounds for the Fr\'echet distance~\cite{DBLP:conf/focs/Bringmann14}, see also~\cite{DBLP:conf/soda/BuchinOS19,DBLP:journals/jocg/BringmannM16}. Continuing this line of research is getting increasingly difficult, although there are still many classical problems from computational geometry without matching lower bounds. In this paper we obtain such bounds for two settings of the classical Hausdorff distance under translation.

\medskip
Besides Hausdorff distance, there are several other distance measures on point sets, including geometric bottleneck matching \cite{efrat_geometry_2001}, Fr\'echet distance~\cite{alt_computing_1995}, and Dynamic Time Warping \cite{DBLP:books/daglib/0019158}.
The geometric bottleneck matching minimizes the maximal distance in a perfect matching between the two given point sets.
Fr\'echet distance and Dynamic Time Warping additionally take the order of the input points into account. 
They both consider the same class of \emph{traversals} of the input points, and the Fréchet distance minimizes the \emph{maximal} distance that occurs during the traversal, while Dynamic Time Warping minimizes the \emph{sum} of distances.

Let us discuss the canonical translational invariant versions of these distance measures. 
For geometric bottleneck matching under translation, Efrat et al.\ designed an $\Ohtilda(n^5)$ algorithm~\cite{efrat_geometry_2001}.
The discrete Fréchet distance under translation has an $\Ohtilda(n^{4.66\dots})$-time algorithm and a conditional lower bound of $n^{4-o(1)}$~\cite{DBLP:conf/soda/BringmannKN19}, see also~\cite{DBLP:conf/esa/BringmannKN20} for algorithm engineering work on this topic. While Dynamic Time Warping is a very popular measure (in particular for video and speech processing), no exact algorithm for its canonical translational invariant version is known in $\ltwo$ since it contains the geometric median problem as a special case \cite{bajaj1988algebraic}.

\medskip
Further work on the Hausdorff distance under translation includes an $\Oh((n+m) \log nm)$-time algorithm for point sets in one dimension~\cite{1dalg}. For generalizations to dimensions $d > 2$ see~\cite{n2alg, chew_geometric_1999}.

\section{Preliminaries} \label{sec:preliminaries}
In this paper we consider finite point sets which lie in $\RR^2$.
For any $p \in \RR^2$, we use $p_x$ and $p_y$ to refer to its first and second component, respectively.
For a point set $A \subset \RR^2$ and a translation $\tau \in \RR^2$, we define $A + \tau \coloneqq \{a + \tau \mid a \in A\}$.
To denote index sets, we often use $[n] \coloneqq \{1, \dots, n\}$.
Given a point $q \in \RR^2$, its $p$-norm is defined as
\[
\norm{q}_p \coloneqq \left(\abs{q_x}^p + \abs{q_y}^p \right)^{\frac{1}{p}}.
\]

We now introduce several distance measures, which are all versions of the famous Hausdorff distance. First, let us define the most basic version.
Let $A, B \subset \mathbb{R}^2$ be two point sets. The \emph{directed Hausdorff distance} is defined as
\[
\dhausdorff(A, B) \coloneqq \max_{a \in A} \min_{b \in B} \norm{a - b}_p.
\]
Note that, intuitively, the directed Hausdorff distance measures the distance from $A$ to $B$ but not from $B$ to $A$, and it is not symmetric. A symmetric variant of the Hausdorff distance, the \emph{undirected Hausdorff distance}, is defined as
\[
	\hausdorff(A,B) \coloneqq \max\{\dhausdorff(A,B), \dhausdorff(B,A)\}.
\]
Note that, by definition, $\dhausdorff(A, B) \leq \hausdorff(A,B)$.
Both of the above distance measures can be modified to a version which is invariant under translation.
The \emph{directed Hausdorff distance under translation} is defined as
\[
	\dhausdorfft(A, B) \coloneqq \min_{\tau \in \RR^2} \dhausdorff(A, B+\tau),
\]
and the \emph{undirected Hausdorff distance under translation} is defined as
\[
	\hausdorfft(A,B) \coloneqq \min_{\tau \in \RR^2} \hausdorff(A,B+\tau).
\]
Again, it holds that $\dhausdorfft(A, B) \leq \hausdorfft(A,B)$.
Naturally, for all of the above distance measures, the decision problem is defined such that we are given two point sets $A, B$ and a threshold distance $\delta$, and ask if the distance of $A, B$ is at most $\delta$.

For the Hausdorff distance on point sets (without translation) the undirected distance is at most as hard as the directed distance, because the undirected distance can be calculated using two calls to an algorithm computing the directed distance.\footnote{Actually, the directed Hausdorff distance is also at most as hard as the undirected Hausdorff distance (thus, they are equally hard), as $\dhausdorff(A,B) = \hausdorff(A \cup B, B)$.} However, note that for the Hausdorff distance under translation, we cannot just compute the directed distance twice and then obtain the undirected distance as we have to take the maximum for the same translation.

 \newcommand{\xtrans}{\ensuremath{[-(n+\frac{1}{2})\epsilon-\epsilon^2, -\frac{3}{2}\epsilon]}}
\newcommand{\ytrans}{\ensuremath{[-\frac{1}{8}, \frac{1}{8}]}}
\newcommand{\trans}{\ensuremath{\xtrans \times \ytrans}}
\newcommand{\Vbar}{\ensuremath{\overline{V}}}
\newcommand{\Vbarr}{\ensuremath{\overline{V_r}}}

\section{\boldmath OV based $(mn)^{1-o(1)}$ lower bound for \lp}\label{sec:ov_lb}

We now present a conditional lower bound of $(mn)^{1-o(1)}$ for the Hausdorff distance under translation --- first for \lone and \linf, and then we discuss how to generalize this bound to \lp. We present the first lower bound only for the \lone case, as the same construction carries over to the \linf case via a rotation of the input sets by $\tfrac{\pi}{4}$.
Our lower bound is based on the hypothesized hardness of the Orthogonal Vectors problem.
\begin{definition}[Orthogonal Vectors Problem (OV)]
	Given two sets $X, Y \subset \{0,1\}^d$ with $|X| = m, |Y| = n$, decide whether there exist $x \in X$ and $y \in Y$ with $x \cdot y = 0$.
\end{definition}
A popular hypothesis from fine-grained complexity theory is as follows.
\begin{definition}[Orthogonal Vectors Hypothesis (OVH)]
	The Orthogonal Vectors problem cannot be solved in time $\Oh((nm)^{1-\epsilon} \text{poly}(d))$ for any $\epsilon > 0$.
\end{definition}
This hypothesis is typically stated and used for the balanced case $n=m$. However, it is known that the hypothesis for the balanced case is equivalent to the hypothesis for any unbalanced case $n = m^\alpha$ for any fixed constant $\alpha > 0$, see, e.g,~\cite[Lemma 2.1]{DBLP:conf/soda/BringmannK18}.

We now describe a reduction from Orthogonal Vectors to Hausdorff distance under translation.
To this end, we are given two sets of $d$-dimensional binary vectors $X = \{x_1, \dots, x_m\}$ and $Y = \{y_1, \dots, y_n\}$ with $\abs{X} = m$ and $\abs{Y} = n$, and we construct an instance of the undirected Hausdorff distance under translation defined by point sets $A$ and $B$ and a decision distance $\delta = 1$. First, we describe the high-level structure of our reduction. The point set $A$ consists only of Vector Gadgets, which encode the vectors of $X$ using $2md$ points. The point set $B$ consists of three types of gadgets:
\begin{itemize}
	\item \emph{Vector Gadgets:} They encode the vectors from $Y$, very similarly to the Vector Gadgets of $A$.
	\item \emph{Translation Gadget:} It restricts the possible translations of the point set $B$.
	\item \emph{Undirected Gadget:} It makes our reduction work for the undirected Hausdorff distance under translation by ensuring that the maximum over the directed Hausdorff distances is always attained by $\dhausdorff(B+\tau, A)$.
\end{itemize}

See Figure \ref{fig:lp_lb_overview} for an overview of the reduction.
Intuitively, the first dimension of the translation chooses the vector $y \in Y$ while the second dimension of the translation chooses the vector $x \in X$. An alignment of the Vector Gadgets within distance 1 is then possible if and only if $x$ and $y$ are orthogonal. Alignments that can circumvent this orthogonality check are not possible as we restrict the translations to a small set of candidates by placing dummy Vector Gadgets on the right side and by including a Translation Gadget.

\subsection{Gadgets}

\begin{figure}
	\centering
	\includegraphics{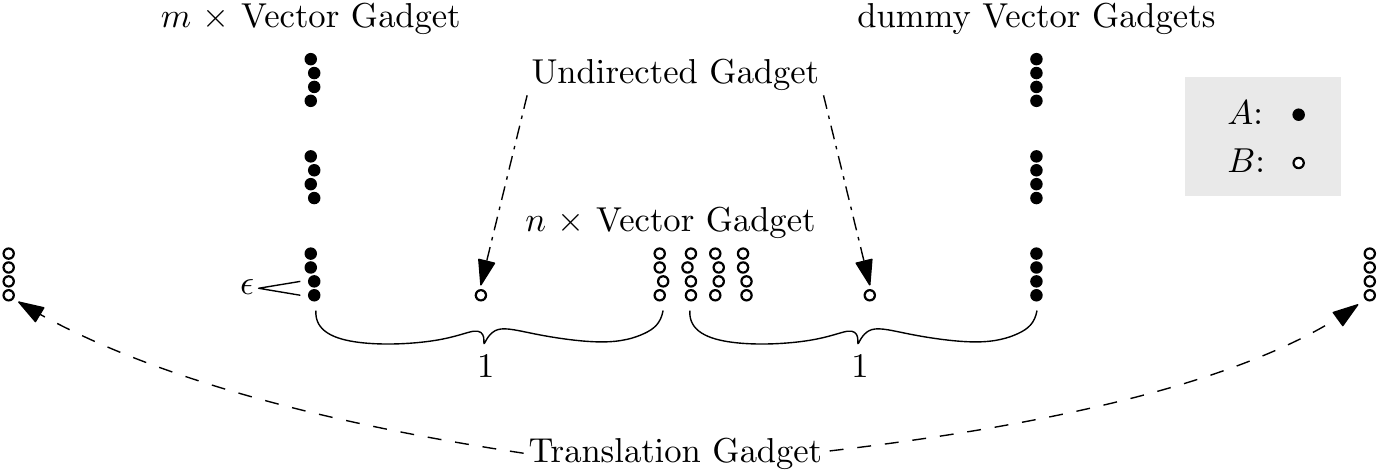}
	\caption{Sketch of the reduction from OV to the undirected Hausdorff distance under translation.}
	\label{fig:lp_lb_overview}
\end{figure}

We now describe the gadgets in detail. Let $\epsilon > 0$ be a sufficiently small constant, e.g., $\epsilon = \frac{1}{20mnd}$. Recall that the distance for which we want to solve the decision problem is $\delta = 1$. Furthermore, we denote the $i$th component of a vector $v$ by $v[i]$ and we use $0^d$ and $1^d$ to denote the $d$-dimensional all-zeros and all-ones vector, respectively.

\paragraph*{Vector Gadget} We define a general Vector Gadget, which we then use at several places by translating it. Given a vector $v \in \{0,1\}^d$, the Vector Gadget consists of the points $p_1, \dots, p_{d} \in \RR^2$:
\[
p_i =
\begin{cases}
	(\epsilon^2, i\epsilon), & \text{if } v[i] = 0 \\
	(0, i\epsilon), & \text{if } v[i] = 1 \\
\end{cases}
\]
We denote the Vector Gadget created from vector $v$ by $V(v)$. Additionally, we define a mirrored version of the gadget $V$ as
\[
	\Vbar(v) \coloneqq V(\bar{v}),
\]
where $\bar{v}$ is the inversion of $v$, i.e., each bit is flipped.

\begin{lemma}\label{lem:vector_gadget}
	Given two vectors $v_1, v_2 \in \{0,1\}^d$ and corresponding Vector Gadgets $V_1 = V(v_1)$ and $V_2 = \Vbar(v_2) + (1, 0)$, we have $\hausdorff(V_1, V_2) \leq 1$ if and only if $v_1 \cdot v_2 = 0$.
\end{lemma}
\begin{proof}
	Let the points of $V_1$ (resp. $V_2$) be denoted as $p_1, \dots, p_d$ (resp. $q_1, \dots, q_d$). First, note that $\norm{p_i - q_j}_1 = 1 + \abs{i-j}\epsilon + (v_1[i]+v_2[j]-1)\epsilon^2 > 1$ for $i \neq j$. Thus, for the Hausdorff distance to be at most $1$, we have to match $p_i$ to $q_i$ for all $i \in [d]$. This is possible if and only if $v_1[i] = 0$ or $v_2[i] = 0$, as $p_i$ and $q_i$ are only at distance larger than 1 for $v_1[i] = 1$ and $v_2[i] = 1$.
\end{proof}
See Figure \ref{fig:vector_gadget} for an example. Note that if we swap both gadgets and invert both vectors (i.e., flip all their bits), the Hausdorff distance does not change and thus an analogous version of Lemma \ref{lem:vector_gadget} holds in this case, as we are just performing a double inversion.
\begin{lemma}\label{lem:vector_gadget_swapped}
Given two vectors $v_1, v_2 \in \{0,1\}^d$ and corresponding Vector Gadgets $V_1 = \Vbar(v_1)$ and $V_2 = V(v_2) + (1, 0)$, we have $\hausdorff(V_1, V_2) \leq 1$ if and only if $\bar{v}_1 \cdot \bar{v}_2 = 0$, where $\bar{v}_1, \bar{v}_2$ are the inversions of $v_1, v_2$.
\end{lemma}

\begin{figure}
	\centering
	\includegraphics{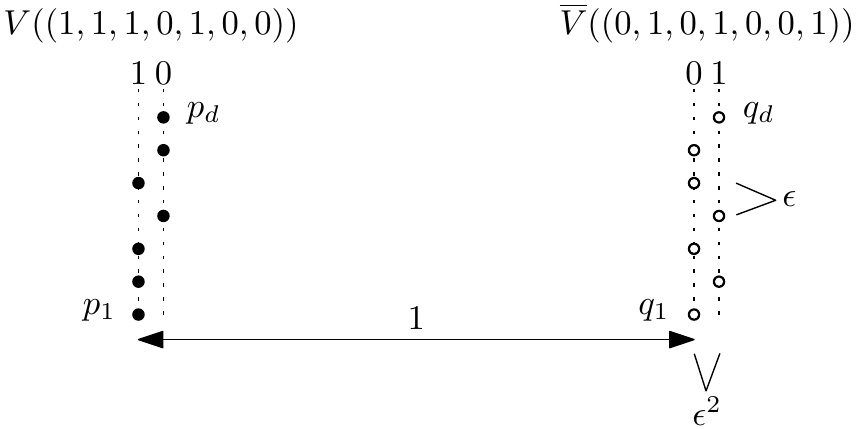}
	\caption{A depiction of the two types of Vector Gadgets and how they are placed to check for orthogonality.}
	\label{fig:vector_gadget}
\end{figure}

For any $x,y,D \in \RR$, we call Vector Gadgets $V_1 = V(v_1) + (x,y)$ and $V_2 = \Vbar(v_2) + (x+D,y)$ \emph{vertically aligned}, or more precisely, \emph{vertically aligned at distance $D$}.

\paragraph*{Translation Gadget} To ensure that $B$ cannot be translated arbitrarily, we introduce a gadget to restrict the translations to a restricted set of candidates.
The Translation Gadget $T$ consists of two translated Vector Gadgets of the zero vector:
\[
	T \coloneqq (\Vbar(1^d) - (2-n\epsilon, 0)) \cup (\Vbar(0^d) + (2+2\epsilon, 0)).
\]

We show that restricting the coordinates of the points of the other set involved in the Hausdorff distance under translation instance, already restricts the feasible translations significantly.
\begin{lemma}\label{lem:translation_gadget}
	Let $P \subset [-1-\frac{1}{2}\epsilon, 1+\frac{1}{2}\epsilon] \times \RR$ be a point set and $T$ the Translation Gadget.
	If $\dhausdorfft(T,P) \leq 1$, then $\tau_x^* \in \xtrans$, where $\tau^*$ is
	any translation satisfying $\dhausdorff(T,P+\tau^*) \leq 1$.
\end{lemma}
\begin{proof}
We show the contrapositive. Therefore, assume the converse, i.e., that $\tau_x^*$ is not contained in  $\xtrans$.
If $\tau_x^* < -(n+\frac{1}{2})\epsilon - \epsilon^2$, then $-1-\frac{1}{2}\epsilon-(-2+n\epsilon+\epsilon^2+\tau_x^*) > 1$ and thus the left part of $T$ cannot contain any point of $P$ at distance at most $1$. If $\tau_x^* > -\frac{3}{2}\epsilon$, then $2+2\epsilon+\tau_x^*-(1+\frac{1}{2}\epsilon) > 1$ and thus the right part of $T$ cannot contain any point of $P$ at distance at most $1$. Thus, $\dhausdorfft(T,P) > 1$.
\end{proof}

\paragraph*{Undirected Gadget}
To ensure that each point in $A$ can be matched to a point in $B$ within distance $1$, we add auxiliary points to $B$. The Undirected Gadget is defined by the point set
\[
	U \coloneqq \{ (-\frac{1}{2}, 0), (\frac{1}{2}, 0) \}.
\]

\begin{lemma} \label{lem:undirected_gadget}
	Given a set of points $P \subset [-1-\frac{1}{2}\epsilon, 1+\frac{1}{2}\epsilon] \times [-\frac{1}{8}, \frac{1}{8}]$, it holds that $\dhausdorff(P, U + \tau) \leq 1$ for any $\tau \in [-(n+\frac{1}{2})\epsilon-\epsilon^2, (n+\frac{1}{2})\epsilon+\epsilon^2] \times \ytrans$.
\end{lemma}
\begin{proof}
By symmetry, we can restrict to proving that the distance of the point set
\[
	P' = P \cap [0, 1+\frac{1}{2}\epsilon] \times \ytrans
\]
to $(\frac{1}{2},0) + \tau$ is at most $1$. For any $p' \in P'$, we have $\abs{p'_x - (\frac{1}{2} + \tau_x)} \leq \frac{1}{2} + (n+\frac{1}{2})\epsilon+\epsilon^2 \leq \frac{1}{2} + \frac{1}{10}$, where the last inequality follows from plugging in $\varepsilon = \frac{1}{20mnd}$, and also $\abs{p'_y - \tau_y} \leq \frac{1}{4}$. Thus, $\norm{p' - ((\frac{1}{2},0) + \tau)}_1 \leq \frac{3}{4} + \frac{1}{10} < 1$.
\end{proof}

\subsection{Reduction and correctness}
We now describe the reduction and prove its correctness. We construct the point sets of our Hausdorff distance under translation instance as follows. The first set, i.e., set $A$, consists only of Vector Gadgets:
\[
	A \coloneqq \left(\bigcup_{i \in [m]} V(x_i) + (-1-\frac{1}{2}\epsilon, i \cdot 2d\epsilon)\right) \cup \left(\bigcup_{i \in [m]} V(1^d) + (1+\frac{1}{2}\epsilon, i \cdot 2d\epsilon)\right)
\]
The second set, i.e., set $B$, consists of Vector Gadgets, the Translation Gadget, and the Undirected Gadget:
\[
	B \coloneqq \left(\bigcup_{j \in [n]} \Vbar(y_j) + (j\epsilon, 0)\right) \cup T \cup U
\]
See Figure \ref{fig:lp_lb_overview} for a sketch of the above construction.
To reference the vector gadgets as they are used in the reduction, we use the notation
\[
	V_r(x_i) \coloneqq V(x_i) + (-1-\frac{1}{2}\epsilon, i \cdot 2d\epsilon) \quad\text{and}\quad \Vbarr(y_j) \coloneqq \Vbar(y_j) + (j\epsilon, 0).
\]

We can now prove correctness of our reduction. In the reduction, we return some canonical positive instance, if the $0^d$ vector is contained in any of the two OV sets. This allows us to drop all $1^d$ vectors from the input, as they cannot be orthogonal to any other vector. Thus, we can assume that all vectors in our input contain at least one 0-entry and at least one 1-entry.

\begin{theorem}\label{thm:ov_hardness}
	Computing the directed or undirected Hausdorff distance under translation in \lone or \linf for two point sets of size $n$ and $m$ in the plane cannot be solved in time $\Oh((mn)^{1-\gamma})$ for any $\gamma > 0$, unless the Orthogonal Vectors Hypothesis fails.
\end{theorem}

\begin{proof}
Recall that we only have to consider the \lone case.
We first prove that there is a pair of orthogonal vectors $x \in X$ and $y \in Y$ if and only if $\hausdorfft(A, B) \leq 1$.
To prove the theorem for the directed and undirected Hausdorff distance under translation at the same time, it suffices to show \enquote{$\Rightarrow$} for the undirected version and \enquote{$\Leftarrow$} for the directed version.

\begin{description}
\item[$\mathbf{\Rightarrow}$:] Assume that there exist $x_i \in X$, $y_j \in Y$ with $x_i \cdot y_j = 0$.
	Then consider the translation $\tau = (-(j+\frac{1}{2})\epsilon, i \cdot 2d\epsilon)$ which vertically aligns the Vector Gadgets $V_r(x_i)$ and $\Vbarr(y_j) + \tau$ at distance $1$.
	As $x_i$ and $y_j$ are orthogonal, it follows from Lemma \ref{lem:vector_gadget} that $\dhausdorff(\Vbarr(y_j)+\tau, A) \leq 1$.
We now show that all of the remaining points of $B + \tau$ have a point of $A$ at distance at most~$1$.
The Vector Gadgets $\Vbarr(y_{j'}) + \tau$ with $j' < j$ are strictly to the left of $\Vbarr(y_j) + \tau$ and are thus also in Hausdorff distance at most $1$ from $V_r(x_i)$.
If $j=n$, then we are done with the Vector Gadgets. Otherwise, consider the Vector Gadget $\Vbarr(y_{j+1}) + \tau$. We claim that each point of it is at distance at most $1$ from $V(1^d) + (1+\frac{1}{2}\epsilon, i \cdot 2d\epsilon)$. As the two gadgets are vertically aligned, we just have to check their horizontal distance, which is
\[
1+\frac{1}{2}\epsilon - ((j+1)\epsilon - (j+\frac{1}{2})\epsilon) = 1.
\]
Thus, by Lemma \ref{lem:vector_gadget}, we have $\dhausdorff(\Vbarr(y_{j+1})+\tau, A) \leq 1$.
Now, by the same argument as above, all gadgets $\Vbarr(y_{j'}) + \tau$ with $j' > j+1$ are in directed Hausdorff distance at most $1$ from $A$.

As the points of the Undirected Gadget $U+\tau$ are closer by a distance of almost $\frac{1}{2}$ to $A$ than the Vector Gadgets in $B+\tau$, also $\dhausdorff(U+\tau, A) \leq 1$ holds.
Finally, we have to show that the Translation Gadget $T+\tau$ is at distance at most $1$ from $A$. As the left part of $T$ and $V_r(x_i)$ are aligned vertically, we only have to check the horizontal distance. The horizontal distance is
\[
-1-\frac{1}{2}\epsilon -(-2+n\epsilon - (j+\frac{1}{2})\epsilon) = 1-(n-j)\epsilon \leq 1
\]
for any $j \in [n]$. Similarly, the distance of the right part of the Translation Gadget from the vertically aligned $V(1^d)$ in $A$ is
\[
	2+2\epsilon-(j+\frac{1}{2})\epsilon-(1+\frac{1}{2}\epsilon) = 1-(j-1)\epsilon \leq 1
\]
for any $j \in [n]$. Thus, by Lemma \ref{lem:vector_gadget} and Lemma \ref{lem:vector_gadget_swapped}, it holds that $\dhausdorff(T+\tau,A) \leq 1$.
As $\tau \in \trans$, we know by Lemma \ref{lem:undirected_gadget} that $\dhausdorff(A, B + \tau) \leq 1$ and thus also $\hausdorfft(A, B) \leq 1$.

\medskip

\item[$\mathbf{\Leftarrow}$:]
Now, assume that $\hausdorfft(A, B) \leq 1$ and let $\tau$ be any translation for which $\dhausdorff(B+\tau, A) \leq 1$. Note that we used the directed Hausdorff distance in the previous statement on purpose, as we prove hardness for both versions. Lemma \ref{lem:translation_gadget} implies that $\tau_x \in \xtrans$.

Let $\Vbarr(y_j) + \tau, \Vbarr(y_{j+1}) + \tau$ be the Vector Gadgets such that $\Vbarr(y_j) + \tau$ has directed Hausdorff distance at most $1$ to the left Vector Gadgets of $A$ and $\Vbarr(y_{j+1}) + \tau$ has directed Hausdorff distance at most $1$ to the right Vector Gadgets of $A$. This is well-defined as the left Vector Gadgets of $A$ and the right Vector Gadgets of $A$ are at distance at least $2+\epsilon-\epsilon^2$ from each other, and thus no Vector Gadget of $B + \tau$ can be at distance at most $1$ from both sides. Furthermore, as $\tau_x \leq -\frac{3}{2}\epsilon$, the Vector Gadget $\Vbarr(y_j) + \tau$ has directed Hausdorff distance at most $1$ to the left Vector Gadgets of $A$, as
\[
j\epsilon-\frac{3}{2}\epsilon - (-1-\frac{1}{2}\epsilon) = 1 + (j-1)\epsilon \leq 1
\]
for $j = 1$. If $j=n$, then $\Vbarr(y_{j+1}) + \tau$ is undefined.

As $\dhausdorff(B+\tau, A) \leq 1$, we know that $\Vbarr(y_j) + \tau$ has directed Hausdorff distance at most $1$ to a gadget $V_r(x)$ for some $x \in X$. We claim that this distance cannot be closer than $1$ as $\Vbarr(y_{j+1}) + \tau$ must have a directed Hausdorff distance at most $1$ from the right side of $A$ or, in case $j=n$, due to the restrictions imposed by the Translation Gadget.
Let us consider the case $j \neq n$ first. Any translation $\tau'$ which places $\Vbarr(y_{j+1})+\tau'$ in directed Hausdorff distance at most $1$ from the right side of $A$ needs to fulfill
\[
1 + \frac{1}{2}\epsilon - ((j+1)\epsilon + \tau'_x) \leq 1
\]
and thus $\tau'_x \geq -(j+\frac{1}{2})\epsilon$, using the fact that each vector in $Y$ contains at least one $0$-entry.
This, on the other hand, implies that $\Vbarr(y_j) + \tau'$ is in Hausdorff distance at least
\[
j\epsilon - (j+\frac{1}{2})\epsilon - (- 1 - \frac{1}{2}\epsilon) = 1
\]
from $V_r(x)$.
Now consider the case $j = n$.
As by Lemma \ref{lem:translation_gadget} we have $\tau_x \geq -(n+\frac{1}{2})\epsilon - \epsilon^2$, it follows that $\Vbarr(y_n) + \tau$ is in Hausdorff distance at least
\[
n\epsilon - (n+\frac{1}{2})\epsilon - (-1-\frac{1}{2}\epsilon) = 1
\]
from $V_r(x)$, using the fact that each vector in $Y$ contains at least one $0$-entry (this is the reason why the $\epsilon^2$ disappears).

By the arguments above, the two gadgets $\Vbarr(y_j) + \tau$ and $V_r(x)$ have to be horizontally aligned as required by Lemma \ref{lem:vector_gadget}. They also have to be vertically aligned as a vertical deviation would incur a Hausdorff distance larger than $1$ for the pair of points in the two gadgets that are in horizontal distance $1$. Then, applying Lemma \ref{lem:vector_gadget}, it follows that $x$ and $y_j$ are orthogonal.
\end{description}

It remains to argue why the above reduction implies the lower bound stated in the theorem. Assume we have an algorithm that computes the Hausdorff distance under translation for \lone or \linf in time $(mn)^{1-\gamma}$ for some $\gamma > 0$. Then, given an Orthogonal Vectors instance $X,Y$ with $\abs{X}=m$ and $\abs{Y}=n$, we can use the described reduction to obtain an equivalent Hausdorff under translation instance with point sets $A, B$ of size $\abs{A} = \Oh(md)$ and $\abs{B} = \Oh(nd)$ and solve it in time $\Oh((mn)^{1-\gamma}\text{poly}(d))$, contradicting the Orthogonal Vectors Hypothesis.
\end{proof}

\subsection{Generalization to \lp}
We can extend the above construction such that it works for all \lp norms with $p \neq \infty$ by changing the spacing between $0$ and $1$ points of the Vector Gadgets and also set $\epsilon$ accordingly.
More precisely, we can set $\epsilon = \frac{1}{40pmnd}$ (instead of $\frac{1}{20mnd}$) and use $\epsilon^{2p}$ as spacing (instead of $\epsilon^2$), i.e., the Vector Gadget for a vector $v \in \{0,1\}^d$ then consists of the points $p_1, \dots, p_{d} \in \RR^2$:
\[
p_i =
\begin{cases}
	(\epsilon^{2p}, i\epsilon), & \text{if } v[i] = 0 \\
	(0, i\epsilon), & \text{if } v[i] = 1 \\
\end{cases}
\]
We prove that these modifications suffice in the remainder of this section.

To this end, first note that in the proof of Theorem \ref{thm:ov_hardness}, the proof for \enquote{$\Rightarrow$} for \lp already follows from the \lone case as the \lone norm is an upper bound on all \lp norms. Thus, we only have to modify the proof of \enquote{$\Leftarrow$}.
To show \enquote{$\Leftarrow$}, note that the only place where we use the \lone norm in the proof is in the invocation of Lemma~\ref{lem:vector_gadget}. Otherwise, we only argue via distances with respect to a single dimension, which carries over to \lp as $\|(x,0)\|_p = |x|$. Thus, we now prove Lemma~\ref{lem:vector_gadget} for the general \lp case.

\begin{proof}[\boldmath Proof of Lemma \ref{lem:vector_gadget} for $L_p$]
To adapt the proof of Lemma \ref{lem:vector_gadget} to the \lp case, we only have to argue that we cannot match any $p_i, q_j$ for $i \neq j$, as the remaining arguments merely argue about distances in a single dimension. We have that
\[
	\norm{p_i - q_j}_p = \left( (\abs{i-j}\epsilon)^p + (1 - (v_1[i]+v_2[j]-1)\epsilon^{2p})^p \right)^{1/p} \geq \left( \epsilon^p + (1 - \epsilon^{2p})^p \right)^{1/p},
\]
which is greater than 1 if $\epsilon^p + (1 - \epsilon^{2p})^p > 1$, which we obtain by using Bernoulli's inequality:
\[
\epsilon^p + (1 - \epsilon^{2p})^p \geq \epsilon^p + 1 - p\epsilon^{2p} \geq 1 + \left(\frac{1}{40pmnd}\right)^p - p\left(\frac{1}{40pmnd}\right)^{2p} > 1.
\]
The remainder of the proof is analogous to the remainder of the proof of Lemma \ref{lem:vector_gadget}.
\end{proof}

By all of the above arguments, the following theorem follows.
\begin{theorem}[Theorem \ref{thm:ov_hardness} for \lp]
Computing the directed or undirected Hausdorff distance under translation in \lp for two point sets of size $n$ and $m$ in the plane cannot be solved in time $\Oh((mn)^{1-\gamma})$ for any $\gamma > 0$, unless the Orthogonal Vectors Hypothesis fails.
\end{theorem}

 \section{\boldmath \threesum based $n^{2-o(1)}$ lower bound for $m \in \Oh(1)$}\label{sec:3sum_lb}
We now present a hardness result for the unbalanced case of the directed and undirected Hausdorff distance under translation.
We base our hardness on another popular hypothesis of fined-grained complexity theory: the \threesum Hypothesis.
Before stating the hypothesis, let us first introduce the \threesum problem.\footnote{Note that we do not explicitly restrict the universe of the integers here. In the WordRAM model, we use the standard assumption that each integer in the input has bit complexity $\Oh(\log n)$. In the RealRAM model, we can perform the common arithmetic operations on reals in constant time, so there is no need to restrict the universe. With these conventions, our reduction works in both models.}
\begin{definition}[\threesum]
	Given three sets of positive integers $X, Y, Z$ all of size $n$, do there exist $x \in X, y \in Y, z \in Z$ such that $x + y = z$?
\end{definition}
The corresponding hardness assumption is the \threesum Hypothesis.
\begin{definition}[\threesum Hypothesis]
	There is no $\Oh(n^{2-\epsilon})$ algorithm for \threesum for any $\epsilon > 0$.
\end{definition}
There are several equivalent variants of the \threesum problem. Most important for us is the convolution \threesum problem, abbreviated as \convthreesum \cite{3sum_to_conv3sum,DBLP:conf/soda/ChanH20}.
\begin{definition}[Conv3SUM]
	Given a sequence of positive integers $X = (x_0, \dots, x_{n-1})$ of size $n$, do there exist $i,j$ such that $x_i + x_j = x_{i+j}$?
\end{definition}
This problem has a trivial $\Oh(n^2)$ algorithm and, assuming the \threesum Hypothesis, this is also optimal up to lower order factors. As \threesum and \convthreesum are equivalent, a lower bound conditional on \convthreesum implies a lower bound conditional on \threesum.

\begin{figure}
	\centering
	\includegraphics[width=\textwidth]{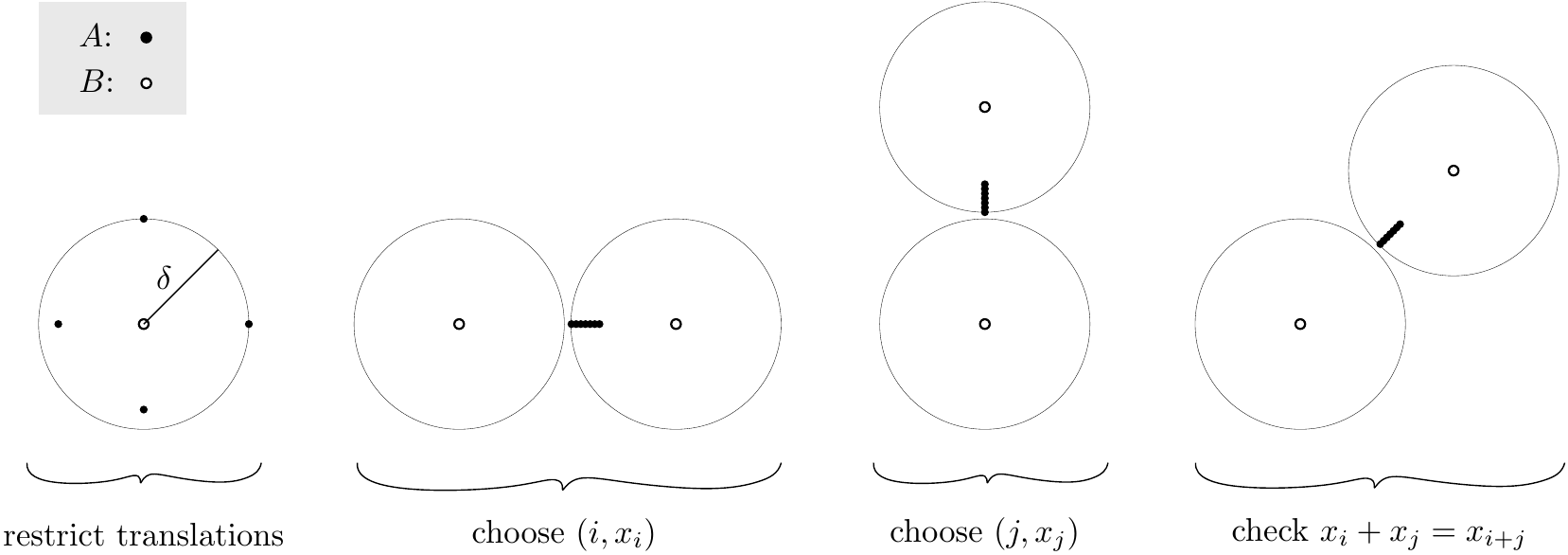}
	\caption{Sketch of the reduction from \convthreesum to the directed and undirected Hausdorff distance under translation in the Euclidean plane.}
	\label{fig:l2_lb_overview}
\end{figure}

Therefore, given a \convthreesum instance defined by the sequence of integers $X$ with $\abs{X} = n$, we create an equivalent instance of the directed Hausdorff distance under translation for $\ltwo$ by constructing two sets of points $A$ and $B$ with $\abs{A} = \Oh(n)$ and $\abs{B} = \Oh(1)$ and providing a decision distance $\delta$.
We provide some intuition for the reduction in the following. See Figure~\ref{fig:l2_lb_overview} for an overview.
Intuitively, we define a low-level gadget from which we build three separate high-level gadgets by rotation and scaling. Recall that in the \convthreesum problem we have to find values $i,j$ which fulfill the equation $x_i + x_j = x_{i+j}$.
Intuitively, we encode the choice of these two values into the two dimensions of the translation: the horizontal translation chooses the pair $(i, x_i)$ in the first high-level gadget and the vertical translation chooses the pair $(j, x_j)$ in the second high-level gadget. The third high-level gadget then allows for a Hausdorff distance below the threshold iff the chosen $i$ and $j$ fulfill the \convthreesum constraint $x_i + x_j = x_{i+j}$. To make this construction also work for the directed Hausdorff distance under translation, we add a simple gadget that restricts translations.
In the remainder of this section, we present the details of our reduction and prove that it implies the claimed lower bound.

\subsection{Construction}

Given a \convthreesum instance with $X \subset [M]$ where $n = \abs{X}$, we now describe the construction of the Hausdorff distance under translation instance with point sets $A, B$ and threshold distance $\delta$.
We use a small enough $\epsilon$, e.g., $\epsilon = (4Mn^2)^{-4}$, as value for microtranslations. Furthermore, we set $\delta = 1 + 4 n^2 \epsilon^2$. The additional $4 n^2 \epsilon^2$ term compensates for the small variations in distance that occur on microtranslations due to the curvature of the $\ltwo$-ball.

\subsubsection{Low-level gadget}

\begin{figure}
	\centering
	\includegraphics{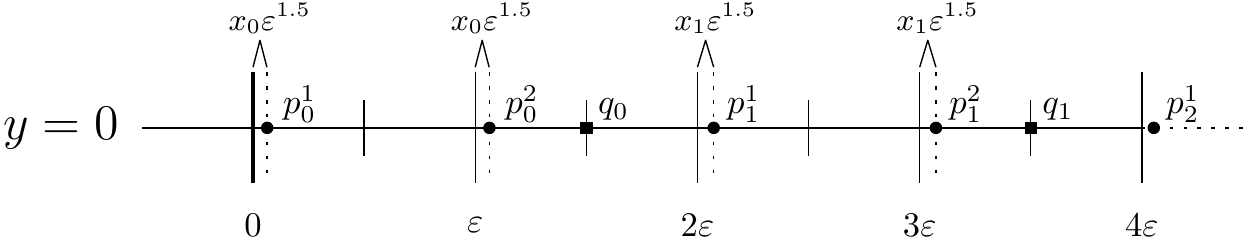}
	\caption{The $A$ set of the low-level gadget of the \threesum reduction, which is used to build the high-level gadgets. We just show the leftmost part of the gadget, but the remainder is similar.}
	\label{fig:low_level_gadget}
\end{figure}

We use a single low-level gadget, which is then scaled and rotated to obtain high-level gadgets. This gadget consists of two point sets $A_l$ and $B_l$. The point set $A_l$ contains what we call \emph{number points} $p_i^1, p_i^2$ and \emph{filling points} $q_i$ for $0 \leq i < n$. The set $B_l$ just contains two points: $r_1$ and $r_2$.
The number points $p_i^1, p_i^2$ encode the number $x_i$, while the filling points make sure that no other translations than the desired ones are possible. See Figure \ref{fig:low_level_gadget} for an overview. All of the points in this gadget are of the form $(x,0)$. The number points are
\[
	p_i^1 = \left(2i\epsilon + x_i \epsilon^{1.5},0\right), \quad p_i^2 = p_i^1 + \left(\epsilon, 0\right)
\]
for $0 \leq i < n$. The filling points are
\[
	q_i = \left(\left(2i + \frac{3}{2}\right)\epsilon, 0\right)
\]
for $0 \leq i < n$.

The points in $B_l$ should introduce a gap to only allow alignment of the number gadgets such that the microtranslations (i.e., those in the order of $\epsilon^{1.5}$) correspond to the number of the gap in the number gadget. To this end, $B_l$ contains the points
\[
	r_1 = (-1, 0),\quad r_2 = (1+\epsilon, 0).
\]

Before we prove properties of the low-level gadget, we first prove that the error due to the curvature of the $\ltwo$-ball is small.
\begin{lemma} \label{lem:error}
	Let $(p_x, p_y), (q_x, q_y) \in \RR^2$ be two points with $\abs{p_x-q_x} \in [\frac{1}{2}, 2]$ and $p_y = q_y$. For any $\tau \in [0,(2n-1)\epsilon]^2$, we have
\[
	\abs{p_x - (q_x + \tau_x)} \leq \norm{p - (q+\tau)}_2 \leq \abs{p_x - (q_x + \tau_x)} + 4  n^2 \epsilon^2.
\]
\end{lemma}
\begin{proof}
As each component is a lower bound for the $\ltwo$ norm, the first inequality follows. Thus, let us prove the second inequality. We first transform
\[
	\norm{p-(q+\tau)}_2 = \sqrt{(p_x-(q_x+\tau_x))^2 + \tau_y^2} = \abs{p_x-(q_x-\tau_x)} \sqrt{1 + \tau_y^2/(p_x-(q_x+\tau_x))^2}.
\]
As $\sqrt{1+x} \le 1+\frac{x}{2}$ for any $x \geq 0$, we have 
\[
\norm{p-(q+\tau)}_2 \le \abs{p_x-(q_x-\tau_x)} + \tau_y^2/(2\abs{p_x-(q_x-\tau_x)}).
\]
As $\tau_y \leq 2(n-1)\epsilon$ and $\abs{p_x-(q_x-\tau_x)} \geq \frac{1}{2}$, we obtain the desired upper bound.
\end{proof}
An analogous statement holds when swapping the $x$ and $y$ coordinates. Note that the $4  n^2 \epsilon^2$ term also occurs in the value of $\delta$ that we chose, as this is how we compensate for these errors in our construction. While we have to consider this error in the following arguments, it should already be conceivable that it will be insignificant due to its magnitude.

To this end, we use a compact notation to denote a value being in a certain range around a value. More concretely, for any $y,r \in \mathbb{R}$, let $x = y \pm r$ denote $x \in [y-r, y+r]$.
We now state two lemmas which show how the Hausdorff distance under translation decision problem is related to the structure of the low-level gadget.

\begin{lemma}\label{lem:stripe_lemma1}
	Given a low-level gadget $A_l, B_l$ as constructed above and the translation being restricted to $\tau \in [0, (2n-1)\epsilon]^2$, it holds that if $\dhausdorff(A_l, B_l+\tau) \leq \delta$, then
\[
	\exists i \in \mathbb{N}: \tau_x = 2i\epsilon + x_i \epsilon^{1.5} \pm 4 n^2 \epsilon^2.
\]
\end{lemma}
\begin{proof}
	Let $\tau \in [0, (2n-1)\epsilon]^2$ and assume $\dhausdorff(A_l, B_l+\tau) \leq \delta$.
	Then all points in $A_l$ are at distance at most $\delta$ from one of the two points in $B_l$.
	Furthermore, both points in $B_l+\tau$ also have at least one close point in $A_l$, as
\[
	\norm{r_1 + \tau - p_0^1}_2 \leq 1 - \tau_x + 4n^2 \epsilon^2 \leq \delta \quad\text{and}\quad \norm{r_2 + \tau - q_{n-1}}_2 \leq 1 + \tau_x - (2n-\frac{3}{2})\epsilon + 4n^2 \epsilon^2 < \delta,
\]
using that $n \geq 1$ and Lemma \ref{lem:error}.

	The gaps between neighboring points in $A_l$ either have width close to $\frac{1}{2}\epsilon$, if the gap is between a number point and a filling point ($p_i^1$ and $q_{i-1}$, or $p_i^2$ and $q_{i}$), or they have a width of $\epsilon$, if the gap is between two number points ($p_i^1$ and $p_i^2$).
Furthermore, the two points in $B_l$ have distance $2 + \epsilon$, so there is an $\epsilon - 8 n^2 \epsilon^2$ gap between their $\delta$-balls.
Thus, there is an $i$ such that $p_i^1$ has distance at most $\delta$ to $r_1$, and $p_i^2$ has distance at most $\delta$ to $r_2$. This alignment of the gadgets can only be realized by a translation $\tau$ for which
\[
\tau_x = 2i\epsilon + x_i \epsilon^{1.5} \pm 4 n^2 \epsilon^2,
\]
which completes the proof.
\end{proof}

\begin{lemma}\label{lem:stripe_lemma2}
	Given a low-level gadget $A_l, B_l$ as constructed above and the translation being restricted to $\tau \in [0, (2n-1)\epsilon]^2$, it holds that if
\[
	\exists i \in \mathbb{N}: \tau_x = 2i\epsilon + x_i \epsilon^{1.5},
\]
then $\hausdorff(A_l, B_l+\tau) \leq \delta$.
\end{lemma}
\begin{proof}
	Let $i \in \mathbb{N}$ and let $\tau_x = 2i\epsilon + x_i \epsilon^{1.5}$. Consider any translations $\tau \in \{\tau_x\} \times [0,2(n-1)\epsilon]$. Due to the restricted translation and Lemma \ref{lem:error}, we can disregard the error terms that arise from the vertical translation $\tau_y$ as they are compensated for by $\delta$.
Then all the points in $A_l$ before and including $p_i^1$ are at distance at most $\delta$ from $r_1 \in B_l+\tau$ and all the points afterwards are at distance at most $\delta$ from $r_2 \in B_l+\tau$. Clearly, both points in $B_l+\tau$ then also have points from $A_l$ at distance $\delta$, and thus $\hausdorff(A_l, B_l+\tau) \leq \delta$.
\end{proof}

\subsubsection{High-level gadgets}

This construction is inspired by the hard instance that was given in \cite{components_lb}. We want to obtain a grid of translations of spacing $\epsilon$ with some microtranslations in the $\Oh(\epsilon^{1.5})$ range. We already defined the low-level gadget above, and we now define the high-level gadgets.

\paragraph*{Column Gadget}
The column gadget induces columns in translational space, i.e., it enforces that valid translations have to lie on one of these columns. The column gadget is actually the low-level gadget we already described above. You can see a sketch of this gadget in Figure \ref{fig:column_gadget}. To semantically distinguish it from the low-level gadget, we refer to the point sets of the column gadget as $A_c$ and $B_c$.

\paragraph*{Row Gadget}
The row gadget induces rows in translational space, i.e., it enforces that valid translations have to lie on one of these rows. We obtain the row gadget by rotating all points in the low-level gadget around the origin by $\pi/2$ counterclockwise. You can see a sketch of this gadget in Figure \ref{fig:row_gadget}. We call the point sets of the row gadget $A_r$ and $B_r$.

\paragraph*{Diagonal Gadget}
The diagonal gadget induces diagonals in translational space, i.e., it enforces that valid translations have to lie on one of these diagonals. As opposed to the column and row gadget, the diagonal gadget also has to be scaled. We scale the sets $A_l$ and $B_l$ separately. We scale $A_l$ such that the gap between the number point pairs $p_i^1, p_i^2$ becomes $\frac{1}{\sqrt{2}}\epsilon$. And we scale $B_l$ such that the gap between the points becomes $2 + \frac{1}{\sqrt{2}}\epsilon$.
After scaling, we rotate the points counterclockwise around the origin by $\pi/4$.
You can see a sketch of this gadget in Figure \ref{fig:diagonal_gadget}.
We call the point sets of the diagonal gadget $A_d$ and $B_d$.

\paragraph*{Translation Gadget}
To restrict the translations for the directed Hausdorff distance under translation, we introduce another gadget. The first set of points $A_t$ contains
\[
	z_l \coloneqq (-1+(2n-1)\epsilon, 0),\quad z_r \coloneqq (1, 0),\quad z_b \coloneqq (0, -1+(2n-1)\epsilon),\quad z_t \coloneqq (0, 1).
\]
The second point set $B_t$ only contains the origin $z_c \coloneqq (0,0)$.
We want to make sure that this gadget behaves well in a certain range.
\begin{lemma} \label{lem:translation_gadget_3sum}
	Given $\tau \in [0, (2n-1)\epsilon]^2$, it holds that $\hausdorff(A_t, B_t + \tau) \leq \delta$.
\end{lemma}
\begin{proof}
As $B_t$ has a point on all sides, clearly $\dhausdorff(B_t + \tau, A_t) \leq \delta$. Furthermore, 
\[
	\norm{z_l - (z_c + \tau)}_2 \leq 1 + 4n^2 \epsilon^2 \leq \delta \quad\text{and}\quad \norm{z_r - (z_c + \tau)}_2 \leq \delta,
\]
using Lemma \ref{lem:error}. Analogous statements hold for $z_b$ and $z_t$. Thus, also $\dhausdorff(A_t, B_t + \tau) \leq \delta$.
\end{proof}

\begin{figure}
\begin{subfigure}[b]{0.3\textwidth}
		\centering
		\includegraphics[scale=0.35]{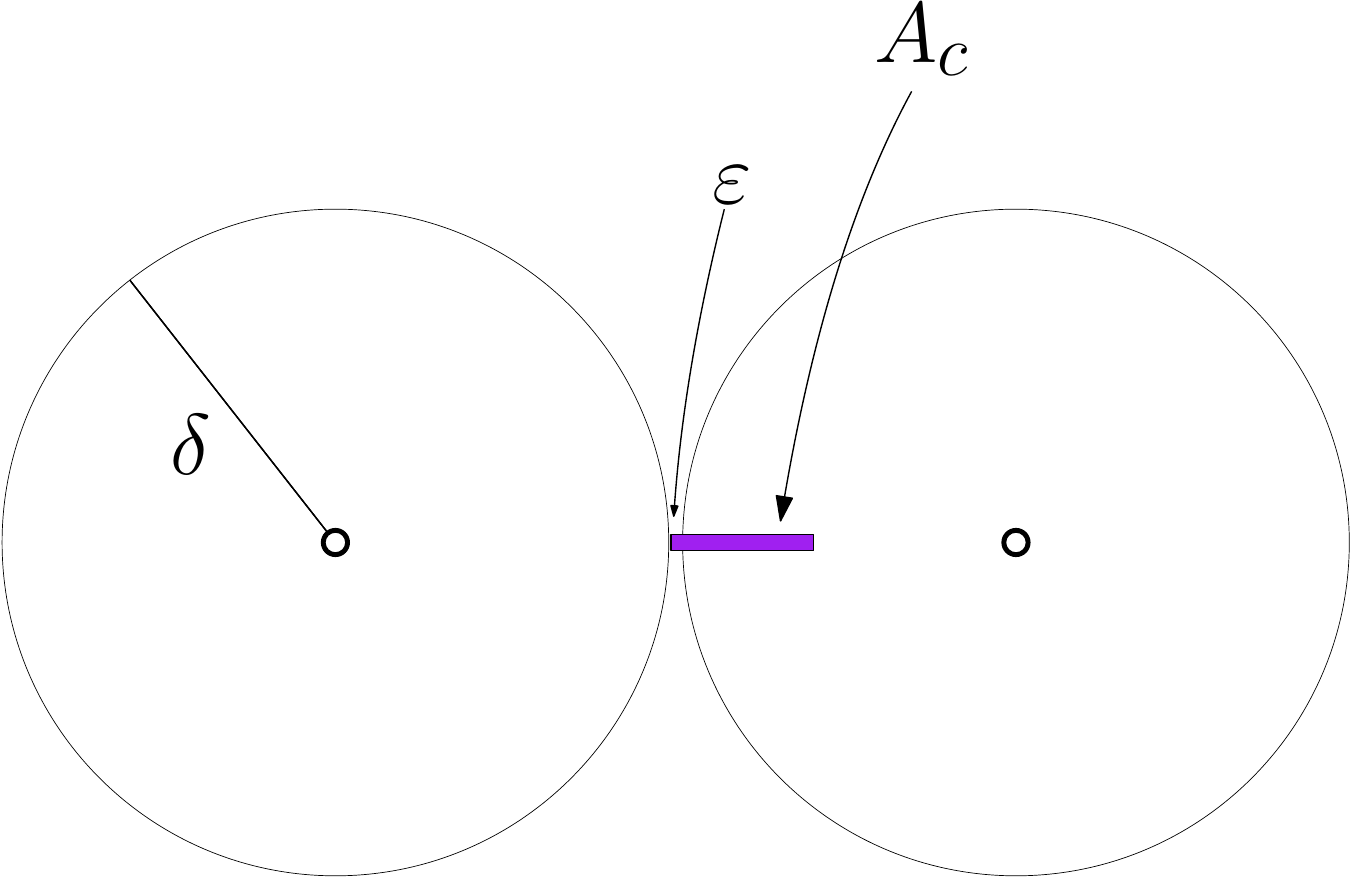}
		\caption{Column Gadget}
		\label{fig:column_gadget}
	\end{subfigure}
	\hfill
	\begin{subfigure}[b]{0.3\textwidth}
		\centering
		\includegraphics[scale=0.35]{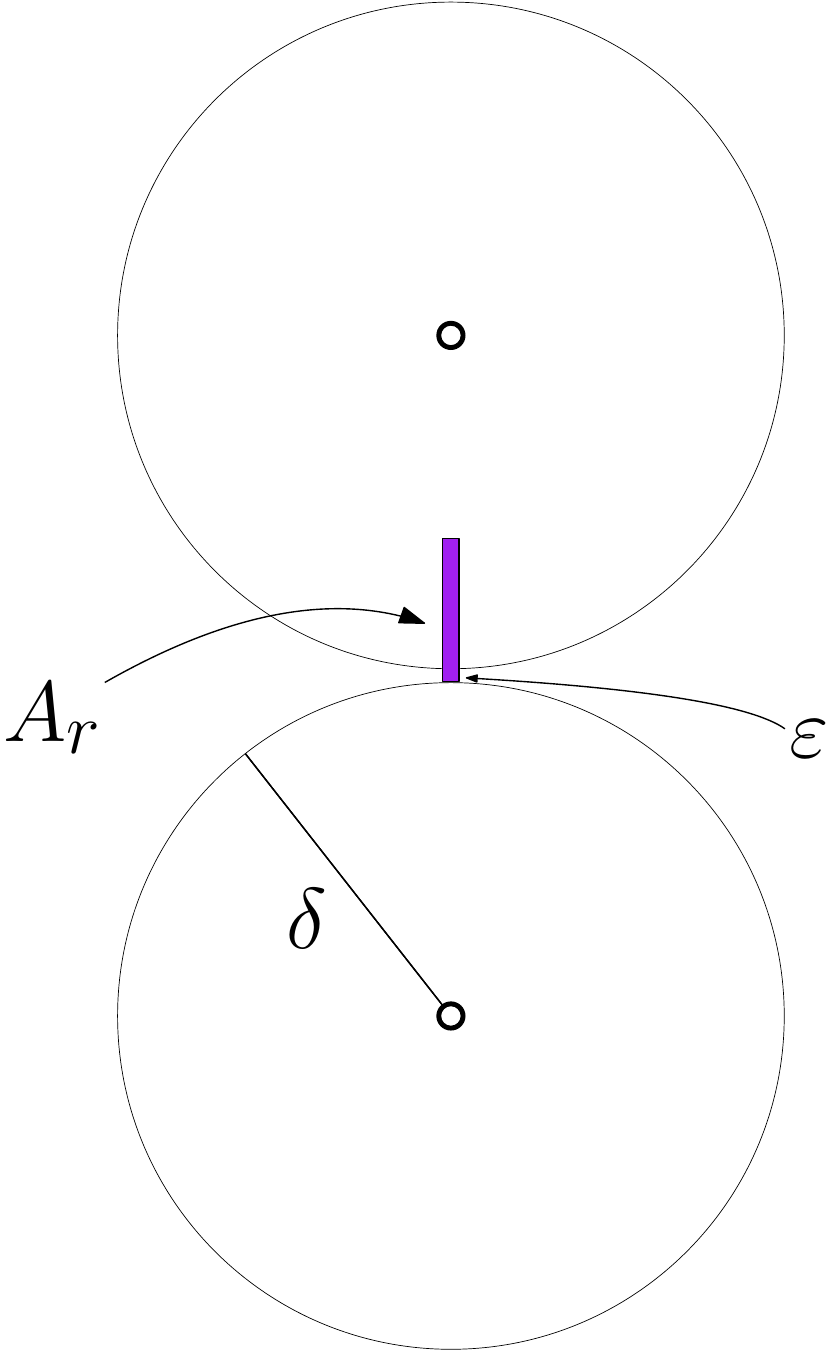}
		\caption{Row Gadget}
		\label{fig:row_gadget}
	\end{subfigure}
	\hfill
	\begin{subfigure}[b]{0.3\textwidth}
		\centering
		\includegraphics[scale=0.35]{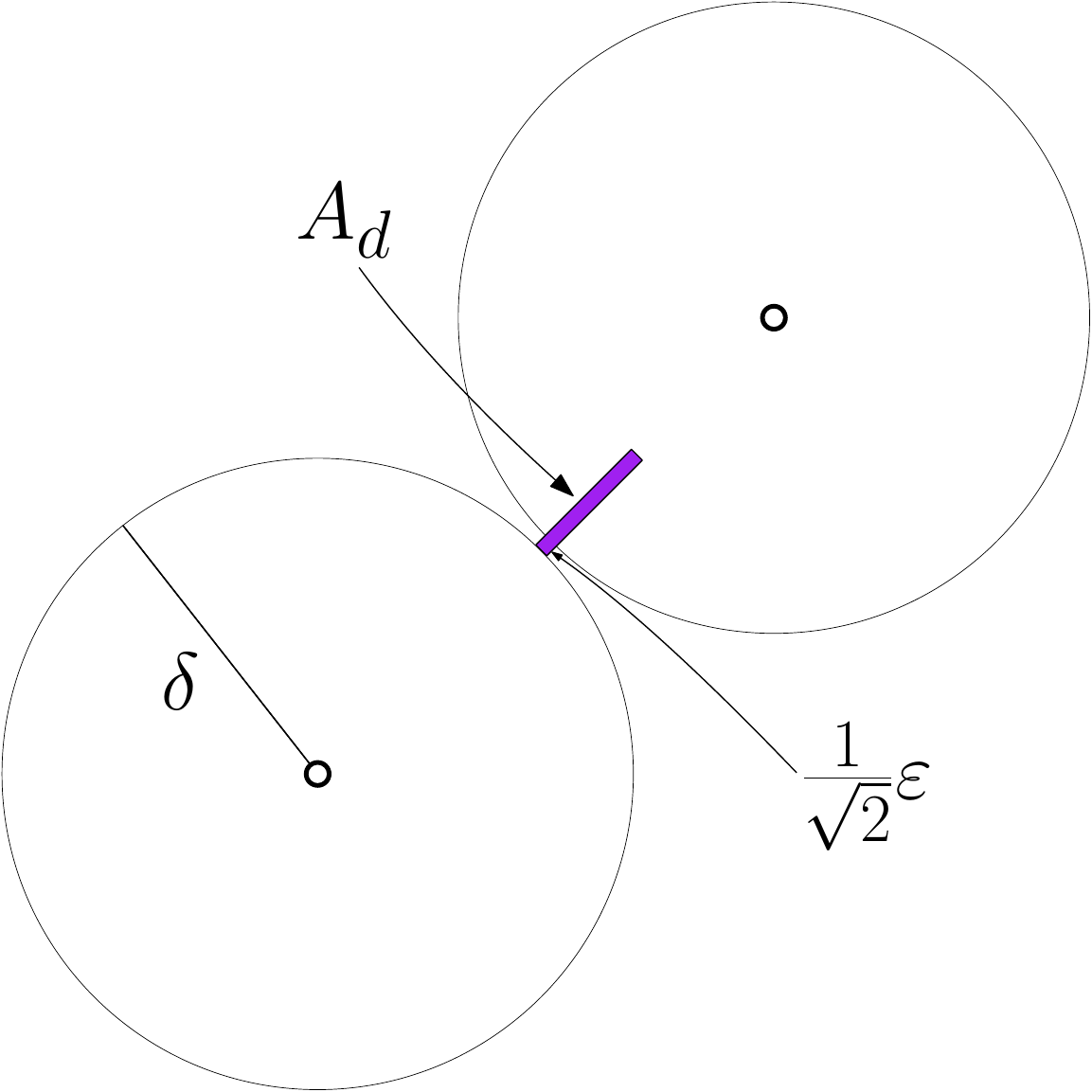}
		\caption{Diagonal Gadget}
		\label{fig:diagonal_gadget}
	\end{subfigure}
	\caption{Three of the high-level gadgets. The points of $A$ are all in the low-level gadgets, while the points in $B$ are explicitly shown including their $\delta$-ball.} \label{fig:main_gadgets}
\end{figure}

\subsubsection{Complete construction}

To obtain the final sets of the reduction, we now place all four described high-level gadgets (i.e., column gadget, row gadget, diagonal gadget, and translation gadget) far enough apart. More explicitly, the point sets $A, B$ of the Hausdorff distance under translation instance are defined as
\[
	A \coloneqq A_c \cup (A_r + (10,0)) \cup (A_d + (20,0)) \cup (A_t + (30,0))
\]
and
\[
	B \coloneqq B_c \cup (B_r + (10,0)) \cup (B_d + (20,0)) \cup (B_t + (30,0)).
\]
The far placement ensures that the two point sets of the respective gadgets have to be matched to each other when the Hausdorff distance under translation is at most delta $\delta$.

\subsection{Proof of correctness}

First, we want to ensure that everything relevant happens in a very small range of translations.

\begin{lemma}\label{lem:translation_constraint}
	Let $\tau \in \RR^2$. If $\dhausdorff(A,B+\tau) \leq \delta$, then $\tau \in [0, (2n-1)\epsilon]^2$.
\end{lemma}
\begin{proof}
Note that for a Hausdorff distance at most $\delta$, the sets $A_c$ and $B_c$ have to matched to each other and analogously for $A_r, B_r$, and $A_d, B_d$, and $A_t, B_t$.
To show the contrapositive, assume $\tau \notin [0, (2n-1)\epsilon]^2$.
For simplicity, we refer to the points in the high-level gadgets with the notation of the low-level gadget.
Due to the translation gadget, we have
\[
	\norm{z_l - (z_c + \tau)}_2 > \delta \quad\text{for}\quad \tau_x > (2n-1)\epsilon + 4n^2\epsilon^2,
\]
and
\[
	\norm{z_r - (z_c + \tau)}_2 > \delta \quad\text{for}\quad \tau_x < -4n^2\epsilon^2.
\]
We now show that under these restricted translations and as $\dhausdorff(A,B+\tau) \leq \delta$, both points $r_1, r_2$ in $B_c$ have at least one point of $A_c$ at distance $\delta$. In the column gadget for $\tau_x \in [-4n^2\epsilon^2,0)$, we have
\[
	\norm{(r_1+\tau) - p_0^1}_2 \geq \abs{- 1 - (p_0^1)_x + \tau_x} > \delta \quad\text{and}\quad \norm{(r_2+\tau) - p_0^1}_2 \geq 1+\epsilon-\Oh(\epsilon^{1.5})  > \delta
\]
for small enough $\epsilon$ and as $x_0 > 0$ and thus there is a component of order $\epsilon^{1.5}$.
On the other hand, for $\tau_x \in ((2n-1)\epsilon, (2n-1)\epsilon + 4n^2\epsilon^2]$, we have
\[
	\norm{r_2+\tau - p_{n-1}^2}_2 \geq 1+\epsilon + \tau_x - (2n-1)\epsilon > \delta \quad\text{and}\quad \norm{r_1+\tau - p_{n-1}^2}_2 \geq 1 + \Oh(\epsilon^{1.5}) - 4n^2\epsilon^2 > \delta
\]
for small enough $\epsilon$. An analogous argument holds for the row gadget and $\tau_y$, as the row gadget is just a rotated version of the column gadget and the translation gadget is symmetric with respect to these gadgets.
\end{proof}

We can now prove the main result of this section.
\begin{theorem}\label{thm:3sum_hardness}
	Computing the directed or undirected Hausdorff distance under translation in \ltwo for two sets of size $n$ and $7$ cannot be solved in time $\Oh(n^{2-\gamma})$ for any $\gamma > 0$, unless the \threesum Hypothesis fails.

\end{theorem}
\begin{proof}

We construct a Hausdorff under translation instance in this proof from a \convthreesum instance as described previously in this section, and then show that they are equivalent.
We first consider how to apply Lemma \ref{lem:stripe_lemma1} and Lemma \ref{lem:stripe_lemma2} to the diagonal gadget.
More precisely, we consider which translations align the gaps of $A_d$ and $B_d$ as is used in these two lemmas.
Consider the constraint $\tau_x = 2k\epsilon + x_k \epsilon^{1.5} \pm 4n^2\varepsilon^2$ that is encoded by the low-level gadget. Recall that we scale this gadget by $\frac{1}{\sqrt{2}}$ and rotate it by $\frac{\pi}{4}$, i.e., we apply the transformation matrix
\[
\frac{1}{\sqrt{2}} \cdot
\begin{pmatrix}
1 & -1 \\
1 & 1
\end{pmatrix}
\cdot
\begin{pmatrix}
\frac{1}{\sqrt{2}} & 0 \\
0 & \frac{1}{\sqrt{2}} \\
\end{pmatrix}
=
\frac{1}{2} \cdot
\begin{pmatrix}
1 & -1 \\
1 & 1
\end{pmatrix}
\]
to the right side of the constraint. Thus, for any $\alpha \in [0, (2n-1)\eps]$, the diagonal gadget encodes the constraints
\[
\begin{pmatrix}
\tau_x \\
\tau_y
\end{pmatrix}
=
\frac{1}{2} \cdot
\begin{pmatrix}
1 & -1 \\
1 & 1
\end{pmatrix}
\cdot
\begin{pmatrix}
2k\epsilon + x_k \epsilon^{1.5} \pm 4n^2\varepsilon^2 \\
\alpha
\end{pmatrix}
=
\frac{1}{2} \cdot
\begin{pmatrix}
2k\epsilon + x_k \epsilon^{1.5} \pm 4n^2\varepsilon^2 & -\alpha \\
2k\epsilon + x_k \epsilon^{1.5} \pm 4n^2\varepsilon^2 & \alpha
\end{pmatrix}
.
\]
By adding up the two constraints, we obtain
\[
	\tau_x + \tau_y = 2k\epsilon + x_k \epsilon^{1.5} \pm 4n^2\varepsilon^2.
\]
We now show correctness of the reduction.
\begin{description}
\item[$\mathbf{\Leftarrow}$:] Assume $X$ is a positive \convthreesum instance.
Then there exist $x_i, x_j$ such that $x_i + x_j = x_{i+j}$.
Consider $\tau = (2i\epsilon + x_i\epsilon^{1.5}, 2j\epsilon + x_j\epsilon^{1.5})$ as translation.
Due to Lemma \ref{lem:stripe_lemma2}, we have that $\hausdorff(A_c, B_c + \tau) \leq \delta$ and analogously $\hausdorff(A_r, B_r + \tau) \leq \delta$.
By the initial observation, we can also apply Lemma \ref{lem:stripe_lemma2} to the diagonal gadget, and thus $\hausdorff(A_d, B_d + \tau) \leq \delta$. Finally, by Lemma \ref{lem:translation_gadget_3sum}, we also have that $\hausdorff(A_t, B_t + \tau) \leq \delta$ for the given $\tau$.

\medskip

\item[$\mathbf{\Rightarrow}$:]
	Assume $\dhausdorfft(A, B) \leq \delta$. From Lemma \ref{lem:translation_constraint}, it follows that $\tau \in [0,(2n-1)\epsilon]^2$.
Then, due to Lemma \ref{lem:stripe_lemma1} and the initial observation about the diagonal gadget, we have that there exist $i,j,k$ that fulfill
\begin{align*}
	\tau_x &= 2i\epsilon + x_i \epsilon^{1.5} \pm 4n^2 \epsilon^2, \\
	\tau_y &= 2j\epsilon + x_j \epsilon^{1.5} \pm 4n^2 \epsilon^2, \\
	\tau_x + \tau_y &= 2k\epsilon + x_k \epsilon^{1.5} \pm 4n^2 \epsilon^2.
\end{align*}
It follows that
\[
	2i\epsilon + x_i \epsilon^{1.5} + 2j\epsilon + x_j \epsilon^{1.5} \pm 8n^2 \epsilon^2 = 2k\epsilon + x_k \epsilon^{1.5} \pm 4n^2 \epsilon^2,
\]
and thus $i+j = k$ and $x_i + x_j = x_k$.
\end{description}

It remains to argue why the above reduction implies the lower bound stated in the theorem. Assume we have an algorithm that computes the Hausdorff distance under translation in \ltwo in time $\Oh(n^{2-\gamma})$ for some $\gamma > 0$. Then, given a \convthreesum instance $X$ with $\abs{X}=n$, we can use the described reduction to obtain an equivalent Hausdorff under translation instance with point sets $A, B$ of size $\abs{A} = \Oh(n)$ and $\abs{B} = 7$ and solve it in time $\Oh(n^{2-\gamma})$, contradicting the \threesum Hypothesis.
\end{proof}

\section{Conclusion}

In this work, we provide matching lower bounds for the running time of two important cases of the fundamental distance measure Hausdorff distance under translation. These lower bounds are based on popular standard hypotheses from fine-grained complexity theory. Interestingly, we use two different hypotheses to show hardness. For the Hausdorff distance under translation for \lp, we show a lower bound of $(nm)^{1-o(1)}$ using the Orthogonal Vectors Hypothesis, while for the imbalanced case of $m = \Oh(1)$ in \ltwo, we show an $n^{2-o(1)}$ lower bound using the \threesum Hypothesis.
We leave it as an open problem whether Hausdorff distance under translation for the balanced case admits a strongly subcubic algorithm or if conditional hardness can be shown.

\bibliographystyle{plainurl}\bibliography{notes}

\begin{thebibliography}{10}

\bibitem{DBLP:conf/focs/AbboudBBK17}
Amir Abboud, Arturs Backurs, Karl Bringmann, and Marvin K{\"{u}}nnemann.
\newblock Fine-grained complexity of analyzing compressed data: Quantifying
  improvements over decompress-and-solve.
\newblock In Chris Umans, editor, {\em 58th {IEEE} Annual Symposium on
  Foundations of Computer Science, {FOCS} 2017, Berkeley, CA, USA, October
  15-17, 2017}, pages 192--203. {IEEE} Computer Society, 2017.
\newblock \href {https://doi.org/10.1109/FOCS.2017.26}
  {\path{doi:10.1109/FOCS.2017.26}}.

\bibitem{DBLP:conf/icalp/AbboudWW14}
Amir Abboud, Virginia~Vassilevska Williams, and Oren Weimann.
\newblock Consequences of faster alignment of sequences.
\newblock In Javier Esparza, Pierre Fraigniaud, Thore Husfeldt, and Elias
  Koutsoupias, editors, {\em Automata, Languages, and Programming - 41st
  International Colloquium, {ICALP} 2014, Copenhagen, Denmark, July 8-11, 2014,
  Proceedings, Part {I}}, volume 8572 of {\em Lecture Notes in Computer
  Science}, pages 39--51. Springer, 2014.
\newblock \href {https://doi.org/10.1007/978-3-662-43948-7\_4}
  {\path{doi:10.1007/978-3-662-43948-7\_4}}.

\bibitem{alt_computing_1995}
Helmut Alt and Michael Godau.
\newblock Computing the {Fréchet} distance between two polygonal curves.
\newblock {\em Int. J. Comput. Geometry Appl.}, 5:75--91, March 1995.
\newblock \href {https://doi.org/10.1142/S0218195995000064}
  {\path{doi:10.1142/S0218195995000064}}.

\bibitem{DBLP:conf/icalp/AmirCLL14}
Amihood Amir, Timothy~M. Chan, Moshe Lewenstein, and Noa Lewenstein.
\newblock On hardness of jumbled indexing.
\newblock In Javier Esparza, Pierre Fraigniaud, Thore Husfeldt, and Elias
  Koutsoupias, editors, {\em Automata, Languages, and Programming - 41st
  International Colloquium, {ICALP} 2014, Copenhagen, Denmark, July 8-11, 2014,
  Proceedings, Part {I}}, volume 8572 of {\em Lecture Notes in Computer
  Science}, pages 114--125. Springer, 2014.
\newblock \href {https://doi.org/10.1007/978-3-662-43948-7\_10}
  {\path{doi:10.1007/978-3-662-43948-7\_10}}.

\bibitem{bajaj1988algebraic}
Chanderjit Bajaj.
\newblock The algebraic degree of geometric optimization problems.
\newblock {\em Discrete \& Computational Geometry}, 3(2):177--191, 1988.

\bibitem{3sum}
Gill Barequet and Sariel Har-Peled.
\newblock Polygon containment and translational in-{H}ausdorff-distance between
  segment sets are {3SUM}-hard.
\newblock {\em International Journal of Computational Geometry \&
  Applications}, 11(04):465--474, August 2001.
\newblock URL:
  \url{https://www.worldscientific.com/doi/abs/10.1142/S0218195901000596},
  \href {https://doi.org/10.1142/S0218195901000596}
  {\path{doi:10.1142/S0218195901000596}}.

\bibitem{DBLP:conf/focs/Bringmann14}
Karl Bringmann.
\newblock Why walking the dog takes time: {F}réchet distance has no strongly
  subquadratic algorithms unless {SETH} fails.
\newblock In {\em 55th {IEEE} Annual Symposium on Foundations of Computer
  Science, {FOCS} 2014, Philadelphia, PA, USA, October 18-21, 2014}, pages
  661--670. {IEEE} Computer Society, 2014.
\newblock \href {https://doi.org/10.1109/FOCS.2014.76}
  {\path{doi:10.1109/FOCS.2014.76}}.

\bibitem{DBLP:conf/soda/BringmannK18}
Karl Bringmann and Marvin K{\"{u}}nnemann.
\newblock Multivariate fine-grained complexity of longest common subsequence.
\newblock In Artur Czumaj, editor, {\em Proceedings of the Twenty-Ninth Annual
  {ACM-SIAM} Symposium on Discrete Algorithms, {SODA} 2018, New Orleans, LA,
  USA, January 7-10, 2018}, pages 1216--1235. {SIAM}, 2018.
\newblock \href {https://doi.org/10.1137/1.9781611975031.79}
  {\path{doi:10.1137/1.9781611975031.79}}.

\bibitem{DBLP:conf/soda/BringmannKN19}
Karl Bringmann, Marvin K{\"{u}}nnemann, and Andr{\'{e}} Nusser.
\newblock Fr{\'{e}}chet distance under translation: Conditional hardness and an
  algorithm via offline dynamic grid reachability.
\newblock In Timothy~M. Chan, editor, {\em Proceedings of the Thirtieth Annual
  {ACM-SIAM} Symposium on Discrete Algorithms, {SODA} 2019, San Diego,
  California, USA, January 6-9, 2019}, pages 2902--2921. {SIAM}, 2019.
\newblock \href {https://doi.org/10.1137/1.9781611975482.180}
  {\path{doi:10.1137/1.9781611975482.180}}.

\bibitem{DBLP:conf/esa/BringmannKN20}
Karl Bringmann, Marvin K{\"{u}}nnemann, and Andr{\'{e}} Nusser.
\newblock When {L}ipschitz walks your dog: Algorithm engineering of the
  discrete {F}r{\'{e}}chet distance under translation.
\newblock In Fabrizio Grandoni, Grzegorz Herman, and Peter Sanders, editors,
  {\em 28th Annual European Symposium on Algorithms, {ESA} 2020, September 7-9,
  2020, Pisa, Italy (Virtual Conference)}, volume 173 of {\em LIPIcs}, pages
  25:1--25:17. Schloss Dagstuhl - Leibniz-Zentrum f{\"{u}}r Informatik, 2020.
\newblock \href {https://doi.org/10.4230/LIPIcs.ESA.2020.25}
  {\path{doi:10.4230/LIPIcs.ESA.2020.25}}.

\bibitem{DBLP:journals/jocg/BringmannM16}
Karl Bringmann and Wolfgang Mulzer.
\newblock Approximability of the discrete {F}r{\'{e}}chet distance.
\newblock {\em J. Comput. Geom.}, 7(2):46--76, 2016.
\newblock \href {https://doi.org/10.20382/jocg.v7i2a4}
  {\path{doi:10.20382/jocg.v7i2a4}}.

\bibitem{DBLP:conf/gis/BuchinDLN19}
Kevin Buchin, Anne Driemel, Natasja van~de L'Isle, and Andr{\'{e}} Nusser.
\newblock klcluster: Center-based clustering of trajectories.
\newblock In Farnoush~Banaei Kashani, Goce Trajcevski, Ralf~Hartmut
  G{\"{u}}ting, Lars Kulik, and Shawn~D. Newsam, editors, {\em Proceedings of
  the 27th {ACM} {SIGSPATIAL} International Conference on Advances in
  Geographic Information Systems, {SIGSPATIAL} 2019, Chicago, IL, USA, November
  5-8, 2019}, pages 496--499. {ACM}, 2019.
\newblock \href {https://doi.org/10.1145/3347146.3359111}
  {\path{doi:10.1145/3347146.3359111}}.

\bibitem{DBLP:conf/soda/BuchinOS19}
Kevin Buchin, Tim Ophelders, and Bettina Speckmann.
\newblock {SETH} says: Weak {F}r{\'{e}}chet distance is faster, but only if it
  is continuous and in one dimension.
\newblock In Timothy~M. Chan, editor, {\em Proceedings of the Thirtieth Annual
  {ACM-SIAM} Symposium on Discrete Algorithms, {SODA} 2019, San Diego,
  California, USA, January 6-9, 2019}, pages 2887--2901. {SIAM}, 2019.
\newblock \href {https://doi.org/10.1137/1.9781611975482.179}
  {\path{doi:10.1137/1.9781611975482.179}}.

\bibitem{DBLP:conf/soda/ChanH20}
Timothy~M. Chan and Qizheng He.
\newblock Reducing {3SUM} to convolution-{3SUM}.
\newblock In Martin Farach{-}Colton and Inge~Li G{\o}rtz, editors, {\em 3rd
  Symposium on Simplicity in Algorithms, SOSA@SODA 2020, Salt Lake City, UT,
  USA, January 6-7, 2020}, pages 1--7. {SIAM}, 2020.
\newblock \href {https://doi.org/10.1137/1.9781611976014.1}
  {\path{doi:10.1137/1.9781611976014.1}}.

\bibitem{chew_geometric_1999}
L.~P. Chew, D.~Dor, A.~Efrat, and K.~Kedem.
\newblock Geometric pattern matching in d-dimensional space.
\newblock {\em Discrete \& Computational Geometry}, 21(2):257--274, February
  1999.
\newblock \href {https://doi.org/10.1007/PL00009420}
  {\path{doi:10.1007/PL00009420}}.

\bibitem{n2alg}
L.~Paul Chew and Klara Kedem.
\newblock Improvements on geometric pattern matching problems.
\newblock In Otto Nurmi and Esko Ukkonen, editors, {\em Algorithm {Theory} —
  {SWAT} '92}, Lecture {Notes} in {Computer} {Science}, pages 318--325.
  Springer Berlin Heidelberg, 1992.

\bibitem{DEBERG2013747}
Mark {de Berg}, Atlas~F. Cook, and Joachim Gudmundsson.
\newblock Fast {F}réchet queries.
\newblock {\em Computational Geometry}, 46(6):747 -- 755, 2013.
\newblock URL:
  \url{http://www.sciencedirect.com/science/article/pii/S0925772112001617},
  \href {https://doi.org/https://doi.org/10.1016/j.comgeo.2012.11.006}
  {\path{doi:https://doi.org/10.1016/j.comgeo.2012.11.006}}.

\bibitem{efrat_geometry_2001}
A.~Efrat, A.~Itai, and M.~J. Katz.
\newblock Geometry helps in bottleneck matching and related problems.
\newblock {\em Algorithmica}, 31(1):1--28, September 2001.
\newblock \href {https://doi.org/10.1007/s00453-001-0016-8}
  {\path{doi:10.1007/s00453-001-0016-8}}.

\bibitem{fedorov2008evaluation}
Andriy Fedorov, Eric Billet, Marcel Prastawa, Guido Gerig, Alireza Radmanesh,
  Simon~K Warfield, Ron Kikinis, and Nikos Chrisochoides.
\newblock Evaluation of brain {MRI} alignment with the robust {H}ausdorff
  distance measures.
\newblock In {\em International Symposium on Visual Computing}, pages 594--603.
  Springer, 2008.

\bibitem{DBLP:journals/comgeo/GajentaanO95}
Anka Gajentaan and Mark~H. Overmars.
\newblock On a class of {$O(n^2)$} problems in computational geometry.
\newblock {\em Comput. Geom.}, 5:165--185, 1995.
\newblock \href {https://doi.org/10.1016/0925-7721(95)00022-2}
  {\path{doi:10.1016/0925-7721(95)00022-2}}.

\bibitem{hausdorff1914grundzuge}
Felix Hausdorff.
\newblock {\em Grundz{\"u}ge der {M}engenlehre}, volume~7.
\newblock von Veit, 1914.

\bibitem{n3alg}
Daniel~P Huttenlocher, Klara Kedem, and Micha Sharir.
\newblock The upper envelope of {V}oronoi surfaces and its applications.
\newblock {\em Discrete \& Computational Geometry}, 9(3):267--291, 1993.

\bibitem{DBLP:journals/jcss/ImpagliazzoPZ01}
Russell Impagliazzo, Ramamohan Paturi, and Francis Zane.
\newblock Which problems have strongly exponential complexity?
\newblock {\em J. Comput. Syst. Sci.}, 63(4):512--530, 2001.
\newblock \href {https://doi.org/10.1006/jcss.2001.1774}
  {\path{doi:10.1006/jcss.2001.1774}}.

\bibitem{king2004survey}
James King.
\newblock A survey of {3SUM}-hard problems.
\newblock 2004.

\bibitem{DBLP:books/daglib/0019158}
Meinard M{\"{u}}ller.
\newblock {\em Information retrieval for music and motion}.
\newblock Springer, 2007.
\newblock \href {https://doi.org/10.1007/978-3-540-74048-3}
  {\path{doi:10.1007/978-3-540-74048-3}}.

\bibitem{3sum_to_conv3sum}
Mihai Patrascu.
\newblock Towards polynomial lower bounds for dynamic problems.
\newblock In {\em Proceedings of the Forty-Second ACM Symposium on Theory of
  Computing}, STOC '10, page 603–610, New York, NY, USA, 2010. Association
  for Computing Machinery.
\newblock \href {https://doi.org/10.1145/1806689.1806772}
  {\path{doi:10.1145/1806689.1806772}}.

\bibitem{1dalg}
Günter Rote.
\newblock Computing the minimum {Hausdorff} distance between two point sets on
  a line under translation.
\newblock {\em Information Processing Letters}, 38(3):123--127, May 1991.
\newblock URL:
  \url{http://www.sciencedirect.com/science/article/pii/0020019091902338},
  \href {https://doi.org/10.1016/0020-0190(91)90233-8}
  {\path{doi:10.1016/0020-0190(91)90233-8}}.

\bibitem{components_lb}
W.~J. Rucklidge.
\newblock Lower bounds for the complexity of the graph of the {Hausdorff}
  distance as a function of transformation.
\newblock {\em Discrete \& Computational Geometry}, 16(2):135--153, February
  1996.
\newblock \href {https://doi.org/10.1007/BF02716804}
  {\path{doi:10.1007/BF02716804}}.

\bibitem{williams2018some}
Virginia Vassilevska~Williams.
\newblock On some fine-grained questions in algorithms and complexity.
\newblock In {\em Proc.\ ICM}, volume~3, pages 3431--3472. World Scientific,
  2018.

\bibitem{DBLP:journals/tcs/Williams05}
Ryan Williams.
\newblock A new algorithm for optimal 2-constraint satisfaction and its
  implications.
\newblock {\em Theor. Comput. Sci.}, 348(2-3):357--365, 2005.
\newblock \href {https://doi.org/10.1016/j.tcs.2005.09.023}
  {\path{doi:10.1016/j.tcs.2005.09.023}}.

\end{thebibliography}

\end{document}